\documentclass[10pt]{article}
\usepackage[utf8]{inputenc}

\usepackage{pdfpages}

\usepackage{graphicx}
\usepackage{color}
\usepackage{algorithm}
\usepackage[noend]{algorithmic}

\usepackage{amsmath,amsthm,amssymb}
\usepackage{latexsym}
\usepackage{array}
\usepackage{stmaryrd}
\usepackage{paralist}
\usepackage{graphicx} 
\usepackage{tikz,authblk}
\usepackage{xspace}
\usepackage{tabulary,booktabs,multirow}
\usepackage{algorithm, algorithmic}
\usepackage{footnote}
\usepackage{subfig}
\usepackage[shortlabels]{enumitem}
\usepackage[toc,page]{appendix}
\setlist{nolistsep}

\usetikzlibrary{positioning,arrows,automata,shapes,decorations.pathreplacing,patterns}

\newtheorem{theorem}{Theorem}[section]
\newtheorem{lemma}[theorem]{Lemma}
\newtheorem{corollary}[theorem]{Corollary}
\newtheorem{proposition}[theorem]{Proposition}

\newcommand{\etal}{\emph{et al.}\xspace}
\newcommand{\ie}{\emph{i.e.}\xspace}
\newcommand{\eg}{\emph{e.g.}\xspace}
\newcommand{\wrt}{\emph{w.r.t.}\xspace}

\newcommand{\RCPlong}{\textsc{Resiliency Checking Problem}\xspace}

 \newcommand{\RDSCPlong}{\textsc{Resiliency Disjoint Set Cover Problem}\xspace}
 \newcommand{\RDSCPshort}{\textsc{RDSCP}\xspace}

\newcommand{\res}{\mathsf{res}}
\newcommand{\Res}{R}

\newcommand{\ur}{\mathit{VR}}

\newcommand{\T}{\mathcal{T}}
\newcommand{\I}{\mathcal{I}}
\renewcommand{\S}{\mathcal{S}}
\newcommand{\F}{\mathcal{F}}
\newcommand{\C}{\mathcal{C}}

\newcommand{\N}{\mathbb{N}}
\newcommand{\Z}{\mathbb{Z}}
\newcommand{\Q}{\mathbb{Q}}
\newcommand{\R}{\mathbb{R}}

\renewcommand{\L}{\mathcal{L}}

\newcommand{\B}{\mathcal{B}}
\DeclareMathOperator{\rank}{{\rm rank}}
\newcommand{\pref}{\ensuremath{\succ}}

\usepackage{makeidx}  
\begin{document}
%
%
%
%
%
\title{Parameterized Resiliency Problems\protect\footnote{Gutin's research was supported by Royal Society Wolfson Research Merit Award. Kouteck\'{y}'s research was supported by a Technion postdoctoral fellowship and
projects 17-09142S of GA \v{C}R. This paper is based on two papers \cite{cramptonAAIM11,CrGuKoWa17} published in conference proceeding of AAIM 2016 and CIAC 2017.}}
%
%
\author[1]{Jason Crampton} 
\author[1]{Gregory Gutin} 
\author[2]{Martin Kouteck{\'y}}
 \author[3]{R\'emi Watrigant}
%
%
%

\affil[1]{Royal Holloway, University of London,
Egham, Surrey, TW20 0EX, UK}
\affil[2]{Technion - Israel Institute of Technology, Haifa, Israel and Charles University, Prague, Czech Republic}
\affil[3]{Universit\'e de Lyon, CNRS, ENS de Lyon, Universit\'e Claude Bernard Lyon 1, LIP, France}


\date{}
\maketitle

\begin{abstract}
We introduce an extension of decision problems called \textit{resiliency problems}. In resiliency problems, the goal is to decide whether an instance remains positive after any (appropriately defined) perturbation has been applied to it. To tackle these kinds of problems, some of which might be of practical interest, we introduce a notion of resiliency for Integer Linear Programs (ILP) and show how to use a result of Eisenbrand and Shmonin (Math. Oper. Res., 2008) on Parametric Linear Programming to prove that \textsc{ILP Resiliency} is fixed-parameter tractable (FPT) under a certain parameterization.  

To demonstrate the utility of our result, we consider natural resiliency versions of several concrete problems, and prove that they are FPT under natural parameterizations. Our first results concern a four-variate problem which generalizes the {\sc Disjoint Set Cover} problem and which is of interest in access control. We obtain a complete parameterized complexity classification for every possible combination of the parameters. Then, we introduce and study a resiliency version of the Closest String problem, for which we extend an FPT result of Gramm et al. (Algorithmica, 2003). We also consider problems in the fields of scheduling and social choice. We believe that many other problems can be tackled by our framework.
\end{abstract}

\section{Introduction}\label{sec:intro}

Questions of integer linear programming (ILP) feasibility are usually answered by finding an integral assignment of variables $x$ satisfying $Ax \le b$. 
By Lenstra's theorem~\cite{Len83}, this problem can be solved in $O^*(f(n)) := O(f(n)L^{O(1)})$ time and space, where $f$ is a function of the number $n$ of variables only, and $L$ is the size of the ILP
(subsequent research has obtained an algorithm of the above running time with $f(n) = n^{O(n)}$ and using polynomial space~\cite{FrTa87,Ka87}).
In the language of parameterized complexity, this means that \textsc{ILP Feasibility} is fixed-parameter tractable (FPT) parameterized by the number of variables. Note that there are a number of parameterized problems for which the only (known) way to prove fixed-parameter tractability is to use Lenstra's theorem\footnote{Lenstra's theorem allows us to prove a mainly classification result, i.e. the FPT algorithm is unlikely to be efficient in practice, nevertheless Lenstra's theorem indicates that efficient FPT algorithms are a possibility, at least for subproblems of the problem under considerations.} \cite{CyFoKoLoMaPiPiSa15}. For more details on this topic, we refer the reader to \cite{CyFoKoLoMaPiPiSa15,DoFe13}.

The notion of \emph{resiliency} measures the extent to which a system can tolerate modifications to its configuration and still satisfy given criteria.
An organization might, for example, wish to know whether it will still be able to continue functioning, even if some of its staff become unavailable. In the language of decision problems, we would like to know whether an instance is still positive after \emph{any} (appropriately defined) modification.
Intuitively, the resiliency version of a problem is likely to be harder than the problem itself; a naive algorithm would consider every allowed modification of the input, and then see whether a solution exists.

In this paper we introduce a framework for dealing with resiliency problems, and study their computational complexity through the lens of fixed-parameter tractability. We define resiliency for Integer Linear Programs (ILP)
by considering a system $\cal R$ of linear inequalities whose variables are partitioned into vectors $x$ and $z$. We denote by ${\cal R}_z$ all inequalities in $\cal R$ containing only variables from $z$.
We say that $\cal R$ is $z$-resilient if for any integral assignment of variables $z$ satisfying all inequalities of ${\cal R}_z$, there is an integral assignment of variables $x$ such that all inequalities of $\cal R$ are satisfied.
We prove in Theorem \ref{thm:ILPres} that the obtained problem can be solved in FPT time for a suitable parameterization, using a result of Eisenbrand and Shmonin on \emph{Parametric Integer Linear Programming}~\cite{EisenbrandS08}.  (Note that the result of Eisenbrand and Shmonin  is an improvement of an earlier result of Kannan \cite{Kannan90}. Unfortunately, Kannan's theorem was based on an incorrect key lemma; the proof of the theorem of Eisenbrand and Shmonin does not have this problem as it uses a correct weak version of the key lemma proved by Eisenbrand and Shmonin, see, e.g., \cite{NguPak17}.)

To illustrate the fact that our approach might be useful in different situations, we apply our framework to several concrete problems.\footnote{Our approach was slightly extended in \cite{KnopKM:2017}, which allowed the authors of \cite{KnopKM:2017} to settle an old conjecture in social choice theory.} Central among them is the \RDSCPlong (\RDSCPshort) defined in Section~\ref{sec:rcp}, which is a generalization of the \textsc{Disjoint Set Cover} problem, and has practical applications in the context of access control~\cite{CrGuWa16-sacmat,LiWaTr09}, where an equivalent formulation of the problem is called the \RCPlong and given in  Section~\ref{sec:rcp}. Informally, given a set $U$ of size $n$ and a family $\cal F$ of $m$ subsets of $U$, the \RDSCPshort asks whether after removing from $\cal F$ any subfamily of size at most $s$ the following property is still satisfied: there are $d$ disjoint set covers of $U$, each of size at most $t$.

 Thus, \RDSCPshort has five natural parameters $n$, $m$, $s$, $d$ and $t$. Since the size of a non-trivial instance can be bounded by a function of $m$ only
 and since in \RCPlong $m$ is usually much larger than the other parameters~\cite{CrGuWa16-sacmat,LiWaTr09}, in
 this paper will only focus on parameters $n$, $s$, $d$ and $t$.
Our main result, obtained using Theorem \ref{thm:ILPres}, is that \RDSCPshort  is FPT when parameterized by $n$.
This, together with some additional results, allow us to determine the complexity of \RDSCPshort  (FPT, XP, W[2]-hard, para-NP-hard or para-coNP-hard) for all combinations\footnote{By definition, a problem with several parameters $p_1,\ldots, p_{\ell}$ is the problem with one parameter, the sum $p_1+\dots +p_{\ell}$.} of the four parameters.
 
We then introduce an extension of the {\sc Closest String} problem, a problem arising in computational biology.
Informally, {\sc Closest String} asks whether there exists a string that is ``sufficiently close'' to each member of a set of input strings.
We modify the problem so that the input strings may be unreliable -- due to transcription errors, for example -- and show that this resiliency version of {\sc Closest String} called {\sc Resiliency Closest String} is FPT when parameterized by the number of input strings.
Our resiliency result on {\sc Closest String} is a generalization of a result of Gramm {\em et al.} for {\sc Closest String} which was proved using Lenstra's theorem~\cite{GrNiRo03}\footnote{Although not being strictly the first problem proved to be FPT using Lenstra's theorem \cite{Seb93}, it is considered as the one which popularized this technique~\cite{CyFoKoLoMaPiPiSa15,DoFe13}.}. 

We also introduce a resiliency version of the scheduling problem of makespan minimization on unrelated machines. We prove that this version is FPT when parameterized by the number of machines, the number of job types and the total expected downtime, generalizing a result of Mnich and Wiese~\cite{MnichW14} provided the jobs processing times are upper-bounded by a number given in unary. 

Finally, we introduce a resilient swap bribery problem in the field of social choice and prove that it is FPT when parameterized by the number of candidates.
 
The remainder of the paper is structured in the following way.
Section~\ref{sec:ILPres} introduces ILP resiliency and proves that it is FPT under a certain parameterization. We then apply our framework to the previously mentioned problems. We establish the fixed-parameter tractability of \RDSCPshort parameterized by $n$ in Section~\ref{sec:rcp} and use it to provide a parameterized complexity classification of the problem. In Section \ref{sec:closeststring}, we introduce a resiliency version of {\sc Closest String} Problem and prove that it is FPT. We study resiliency versions of scheduling and social choice problems in Sections \ref{sec:scheduling} and \ref{sec:bribery}. We conclude the paper in Section~\ref{sec:dis}, where we discuss related literature.

\section{ILP resiliency}\label{sec:ILPres}
Recall that questions of ILP feasibility are typically answered by finding an integral assignment of variables $x$ satisfying $Ax \le b$. 
Let us introduce resiliency for ILP as follows. We add another set of variables $z$, which can be seen as ``resiliency variables''. We then consider the following ILP\footnote{To save space, we will always implicitly assume that integrality constraints are part of every ILP of this paper.} 
denoted by $\cal R$:
\begin{eqnarray}
Ax & \le & b  \label{eqn:var-x}\\
Cx + Dz & \le & e \label{eqn:var-x-z}\\
Fz & \le & g \label{eqn:var-z}
\end{eqnarray}
The idea is that inequalities (\ref{eqn:var-x}) and (\ref{eqn:var-x-z}) represent the intrinsic structure of the problem, among which inequalities (\ref{eqn:var-x-z}) represent how the resiliency variables modify the instance. Inequalities (\ref{eqn:var-z}), finally, represent the structure of the resiliency part.
The goal of \textsc{ILP Resiliency} is to decide whether $\cal R$ is $z$-{\em resilient}, i.e. whether for \emph{any} integral assignment of variables $z$ satisfying inequalities (\ref{eqn:var-z}), there exists an integral assignment of variables $x$ satisfying (\ref{eqn:var-x}) and (\ref{eqn:var-x-z}). 

In $\cal R$, we will assume that all entries of matrices in the left hand sides and vectors in the right hand sides are rational numbers. The dimensions of the vectors $x$ and $z$ will be denoted by $n$ and $p$, respectively, and the total number of rows in $A$ and $C$ will be denoted by $m$. Let $\kappa({\cal R}):=n+p+m$.

Our main result establishes that \textsc{ILP Resiliency} is FPT when parameterized by $\kappa({\cal R})$, provided that part of the input is given in unary. 
Our method offers a generic framework to capture many situations. Firstly, it applies to ILP, a general and powerful model for representing many combinatorial problems. Secondly, the resiliency part of each problem can be represented as a whole ILP with its own variables and constraints, instead of, say, a simple additive term. Hence, we believe that our method can be applied to many other problems, as well as many different and intricate definitions of resiliency. 

To prove our main result we will use the work of Eisenbrand and Shmonin \cite{EisenbrandS08}. A {\em polyhedron} $P$ is a set of vectors of the form $P=\{x\in \mathbb{R}^{\nu}:\ Ax\le b\}$ for some matrix $A\in \mathbb{R}^{\mu\times \nu}$ and some vector $b\in \mathbb{R}^{\mu}.$ The polyhedron is {\em rational} if $A\in \mathbb{Q}^{\mu\times \nu}$ and $b\in \mathbb{Q}^{\mu}.$
For a rational polyhedron $Q \subseteq \R^{m+p}$, define $ Q/\Z^p := \{h \in \mathbb{Q}^m :\ (h,\alpha) \in Q \textrm{ for some } \alpha \in \mathbb{Z}^p\}$.
The \textsc{Parametric Integer Linear Programming} (\textsc{PILP}) problem takes as input a rational matrix $J \in \Q^{m \times n}$ and a rational polyhedron $Q \subseteq \R^{m+p}$, and asks whether the following expression is true:
\begin{equation*}
\forall h \in Q/\Z^p \quad \exists x \in \Z^n: \quad Jx \leq h \label{eqn:es_sentence}
\end{equation*}
Eisenbrand and Shmonin \cite[Theorem 4.2]{EisenbrandS08} proved that \textsc{PILP} is solvable in polynomial time if the number of variables $n+p$ is fixed. From this result, an interesting question is whether this running time is a uniform or non-uniform polynomial algorithm~\cite{DoFe13}, and in particular for which parameters one can obtain an FPT algorithm. By looking more closely at their algorithm, one can actually obtain the following result:

\begin{theorem}
\label{thm:es_pilp}
\textsc{PILP} can be solved in time $O^*(f(n, p)\varphi^{g_1(n, p)}m^{g_2(n, p)})$, where $\varphi\ge 2$ is an upper bound on the encoding length of entries of $J$ and $f$, $g_1$ and $g_2$ are some computable functions.
\end{theorem}

\noindent
\textbf{Complexity remark.} Let us first mention that \textsc{PILP} belongs to the second level of the polynomial hierarchy, and is $\Pi_2^P$-complete~\cite{EisenbrandS08}. Secondly, the polyhedron $Q$ in Theorem \ref{thm:es_pilp} can be viewed as being defined by a system $Rh + S \alpha \le t$, where $h \in \R^m$ and $\alpha \in \Z^p$. 
\begin{corollary}\label{cor:es_pilp}
If all entries of $J$ are given in unary, then \textsc{PILP} is FPT when parameterized by $(n, m, p)$.
\end{corollary}
\begin{proof}
We may assume that there is an upper bound $N\ge 2$ on the absolute values of entries of $J$ and $N$ is given in unary. Thus, the running time of the algorithm of Theorem \ref{thm:es_pilp} is $O^*(h(n,m, p)(\log N)^{F})$, where $F=f(n,m,p)$ and $h$ is some computable function. 

It was shown in \cite{chitnisTA11} that $(\log N)^{F}\le (2F\log F)^F + N/2^F,$ which concludes the proof.
\end{proof}

We now prove the main result of our framework, which will be applied in the next sections to concrete problems.

\begin{theorem}\label{thm:ILPres}
\textsc{ILP Resiliency} is FPT when parameterized by $\kappa({\cal R})$ provided the entries of matrices $A$ and $C$ are given in unary. 
\end{theorem}

\begin{proof}
We will reduce \textsc{ILP Resiliency} to \textsc{Parametric Integer Linear Programming}.
Let us first define $J$ and $Q$. 
Let $h=(h^1, h^2)$ with $h^1$ and $h^2$ being $m_1$ and $m_2$ dimensional vectors, respectively. Then the polyhedron $Q$ is defined as follows: $h^1 =b, h^2 = e - D\alpha , F\alpha \leq g.$
Furthermore, $J$ is defined as: $Ax \leq h^1,\ Cx\leq h^2.$


Recall that $h^1 = b$ and $h^2 = e-D\alpha$ and $\alpha$ satisfies $F\alpha \leq g$, so for all $h \in Q/\Z^p$ there exists an integral $x$ satisfying the above if and only if for all $z$ satisfying $Fz \leq g$, there is an integral $x$ satisfying (\ref{eqn:var-x}) and (\ref{eqn:var-x-z}). Moreover, the dimension of $x$ is $n$, the integer dimension of $Q$ is $p$ and the number of inequalities of $J$ is $m_1 + m_2 = m$, so applying Corollary~\ref{cor:es_pilp} indeed yields the required FPT algorithm.
\end{proof}

\section{Resiliency Disjoint Set Cover}\label{sec:rcp}

The \textsc{Set Cover} problem is one of the classical NP-hard problems \cite{GaJo79}. Its input comprises a finite set $U$ called the \textit{universe}, a family\footnote{We use the term {\em family} as a synonym of {\em multiset}, i.e. a family may have multiple copies of the same element. The operations of {\em union}, {\em intersection} and {\em deletion} on pairs ${\cal F}, {\cal F}'$ of families are defined in the natural way using $\max\{p,p'\}$, $\min\{p,p'\}$ and $\max\{0,p-p'\},$ where $p$ and $p'$ are the numbers of copies of the same element in ${\cal F}$ and ${\cal F}'$, respectively.} $\mathcal{F}$ of $m$ subsets of $U$,  the \RDSCPshort asks whether for every subfamily ${\cal S} \subseteq \mathcal{F}$ with $|{\cal S}| \le s$, one can find ${\cal T}_1, \dots, {\cal T}_d \subseteq \mathcal{F} \setminus {\cal S}$, such that for every $i, j\in [d]$: (i) ${\cal T}_i \cap {\cal T}_j = \emptyset$ whenever $i \neq j$ (ii) $|{\cal T}_i| \le t$, and (iii) $\bigcup_{X \in {\cal T}_i} X = U$.  $\mathcal{F}$ of subsets of $U$, and an integer $t$. It asks whether there is a subfamily ${\cal T} \subseteq \mathcal{F}$ of cardinality at most $t$ such that $\cup_{X \in {\cal T}} X = U$ (such a subfamily is called a \textit{set cover}). We may assume that $t\le |U|$ since every set in a minimal set cover $C$ must have an element of $U$ not contained in any other set of $C$. 
A natural generalization of \textsc{Set Cover}  is the \textsc{Disjoint Set Cover} problem which takes an additional parameter $d$, and asks for the existence of $d$ disjoint set covers, each of cardinality at most $t$.\footnote{\textsc{Disjoint Set Cover} was previously introduced in the literature \cite{PananjadyBV15,PananjadyBV17} in a different way: find the maximum number of disjoint (arbitrary sized) set covers. However, unlike our formulation, the previously introduced one does not extend the {\sc Set Cover} problem, where the set cover is required to be of bounded size.}

In the resiliency version of \textsc{Disjoint Set Cover} studied in this paper, one is given an integer $s$, and asks whether after the removal of any subfamily ${\cal S} \subseteq \mathcal{F}$ with $|{\cal S}| \leq s$, one still can find $d$ disjoint set covers, each of size at most $t$ (and disjoint from $\cal S$). More formally, we have the following:

\begin{center}
\fbox{
\begin{minipage}[c]{.9\textwidth}
\RDSCPlong (\RDSCPshort)\\
\underline{Input:} A universe $U$, a multiset $\mathcal{F}$ of subsets of $U$, integers $s$, $d$ and $t$.\\
\underline{Question:} For every ${\cal S} \subseteq \mathcal{F}$ such that $|{\cal S}|\le s$, do there exist ${\cal T}_1, \dots, {\cal T}_d \subseteq \mathcal{F} \setminus {\cal S}$ such that for every $i,j \in [d]$, we have $|{\cal T}_i| \leq t$, $\cup_{X \in {\cal T}_i} X = U$, and ${\cal T}_i \cap {\cal T}_j = \emptyset$ for all $j \neq i$?
\end{minipage}}~\\ 
\end{center}
In this formulation and elsewhere, for an integer $p$, $[p]$ stands for the set $\{1,2,\dots ,p\}.$

The motivation for the study of \textsc{Resiliency Disjoint Set Cover} comes from a problem arising in access control.

\subsection{Application of \RDSCPshort}

Access control is an important topic in computer security, and deals with the idea of enforcing a policy that specifies which users are authorized to access a given set of resources.
Authorization policies are frequently augmented by additional policies, articulating concerns such as separation of duty and resiliency.
The \RCPlong was introduced by Li {\em et al.}~\cite{LiWaTr09} and asks whether it is always possible to form teams of users, the members of each team collectively having access to all resources, even if some users are unavailable.

Given a set of users $V$ and set of resources $R$, an \emph{authorization policy} is a relation $\ur \subseteq V \times \Res$; we say $v\in V$ is \emph{authorized} for resource $r$ if $(v,r) \in \ur$.
Given an authorization policy $\ur \subseteq V \times \Res$, an instance of the \RCPlong is defined by a resiliency policy $\res(P, s, d, t)$, where $P \subseteq \Res$, $s \ge 0$, $d \ge 1$ and $t \ge 1$. We say that $\ur$ \emph{satisfies} $\res(P, s, d, t)$ if and only if for every subset $S \subseteq V$ of at most $s$ users, there exist $d$ pairwise disjoint subsets of users $V_1, \dots, V_d$ such that for all $i \in [d]$:
\begin{align}
&V_i \cap S = \emptyset, \label{def:cond:intersect} \\
&|V_i| \le t \text{ and } \bigcup_{v \in V_i} \{r \in \Res \text{ s.t. } (v, r) \in \ur \} \supseteq P. \label{def:cond:size}
\end{align}
In other words, $\ur$ satisfies $\res(P,s,d,t)$ if we can find $d$ disjoint groups of users, even if up to $s$ users are unavailable, such that each group contains no more than $t$ users and the users in each group are collectively authorized for the resources in $P$.

Observe that there is an immediate reduction from \RCPlong to \textsc{Resiliency Disjoint Set Cover}: the elements of the universe are the resources, and the sets are in one-to-one correspondence with the users, \ie a set contains all resources a given user has access to.

Observe that \RDSCPshort has five natural parameters: $n=|U|$, $m=|\mathcal{F}|$, $s$, $d$, $t$. As explained earlier, we will only focus on parameters $n, s, d$ and $t$. This choice is also motivated by our application, since it is frequently assumed that the number of users of an organization is usually much larger than the set of resources \cite{LiWaTr09}.

\subsection{Parameterized Complexity Classification of   \RDSCPshort}

As said before, \RDSCPshort contains several natural parameters, namely $n$, $s$, $d$ and $t$. In the next two subsections we prove the results leading to the classification of parameterized complexity of \RDSCPshort shown in Fig. \ref{fig:lattices}. An arrow $A \longrightarrow B$ means that $A$ is a larger parameter than $B$, in the sense that the existence of an FPT algorithm parameterized by $B$ implies the existence of an FPT algorithm parameterized by $A$, and, conversely, any negative result parameterized by $A$ implies the same negative result parameterized by $B$.

The FPT results follow from Theorem \ref{thm:fpt-param-by-p-only} which is
 proved in Section \ref{sec:fpt-part}. The remaining results are obtained in Section \ref{sec:other}.

\begin{figure}[!h]
  \begin{center}
    
\begin{tikzpicture}%
  [fpt/.style={rectangle,draw,minimum width=1cm},%
   xp/.style={rectangle,draw,pattern=north west lines,pattern color=black!60,minimum width=1cm},%
   paranphard/.style={rectangle,draw,fill=black!30,minimum width=1cm},%
   unknown/.style={circle,draw}]
    
   \node[fpt, minimum width=12pt,label=right:FPT] (legend-fpt) at (0,-6) {};
  \node[xp, minimum width=12pt,label=right:{W[2]-hard but XP}] (legend-xp) at (2,-6) {};
  \node[paranphard, minimum width=12pt,label=right:para-(co)NP-hard] (legend-para) at (6,-6) {};

  \node[fpt] (psdt) at (3.75,1) {$n,s,d,t$};

  \node[fpt] at (1,-1) (pst) {$n,s,t$};
  \node[fpt] at (2.75,-1) (psd) {$n,s,d$};
  \node[fpt] at (4.5,-1) (pdt) {$n,d,t$};
  \node[xp] at (6.25,-1) (sdt) {$s,d,t$};

  \node[fpt] (pt) at (0.5,-3) (pt) {$n,t$};
  \node[fpt] at (1.8,-3) (ps) {$n,s$};
  \node[fpt] at (3.1,-3) (pd) {$n,d$};
  \node[paranphard] (sd) at (4.4,-3) {$s,d$};
  \node[paranphard] (st) at (5.7,-3) {$s,t$};
  \node[paranphard] (dt) at (7,-3) {$d,t$};

  \node[fpt] (p) at (1.5,-5) (p) {$n$};
  \node[paranphard] (s) at (3,-5) (s) {$s$};
  \node[paranphard] (t) at (4.5,-5) (t) {$t$};
  \node[paranphard] (d) at (6,-5) (d) {$d$};
  \draw[->] (psdt) -- (psd);
  \draw[->] (psdt) -- (pst);
  \draw[->] (psdt) -- (pdt);
  \draw[->] (psdt) -- (sdt);
  \draw[->] (pdt) -- (pd);
  \draw[->] (pdt) -- (pt);
  \draw[->] (pdt) -- (dt);
  \draw[->] (pst) -- (ps);
  \draw[->] (pst) -- (pt);
  \draw[->] (pst) -- (st);
  \draw[->] (psd) -- (ps);
  \draw[->] (psd) -- (pd);
  \draw[->] (psd) -- (sd);
  \draw[->] (sdt) -- (sd);
  \draw[->] (sdt) -- (st);
  \draw[->] (sdt) -- (dt);
  \draw[->] (pd) -- (p);
  \draw[->] (pd) -- (d);
  \draw[->] (dt) -- (d);
  \draw[->] (dt) -- (t);
  \draw[->] (pt) -- (p);
  \draw[->] (pt) -- (t);
  \draw[->] (ps) -- (p);
  \draw[->] (ps) -- (s);
  \draw[->] (pd) -- (p);
  \draw[->] (pd) -- (d);
  \draw[->] (sd) -- (s);
  \draw[->] (sd) -- (d);
  \draw[->] (st) -- (s);
  \draw[->] (st) -- (t);
  \draw[->] (dt) -- (d);
  \draw[->] (dt) -- (t);
  
 \end{tikzpicture}
 
    \caption{Schema of the parameterized complexity of  \RDSCPshort.}
    \label{fig:lattices}
  \end{center}
\end{figure}
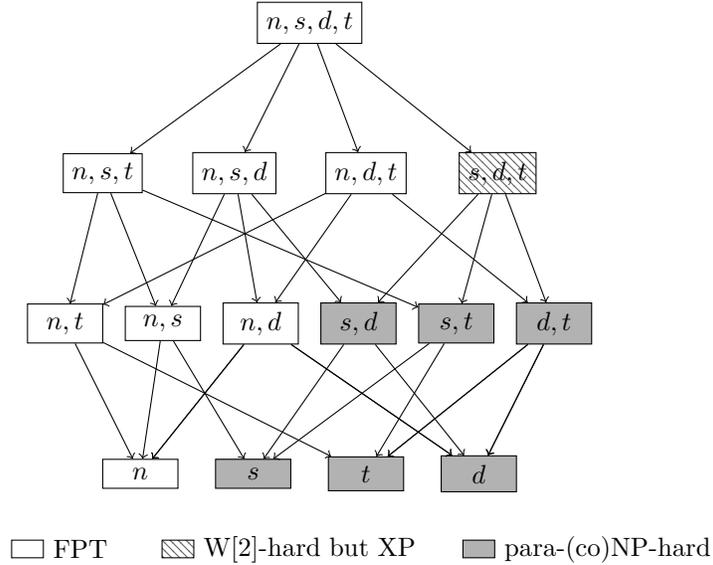

\subsubsection{Fixed-Parameter Tractability of \RDSCPshort}\label{sec:fpt-part}
 
In this subsection we prove that \RDSCPshort is FPT parameterized by $n$.
We first introduce some notation.
In the following, $U, \mathcal{F}, s, d, t$ will denote an input of \RDSCPshort, as defined previously.
For all $N \subseteq U$, let ${\cal F}_N = \{X \in {\cal F}: X = N\}$ (notice that we may have $\F_N = \emptyset$ for some $N \subseteq U$).

Roughly speaking, the idea is that in order to construct several disjoint set covers, it is sufficient to know how many sets were picked from $\F_N$, for every $N \subseteq U$ (observe that we may assume that a single set cover does not contain more than one set from each $\F_N$).
We first define the set of all possible set covers of the instance:
\[
\C = \left\{\{N_1, \dots, N_b\} : b \le t, N_i \subseteq U, i \in [b], \bigcup_{i=1}^b N_i = U\right\}.  
\]
Then, for any $N \subseteq U$, we denote the set of set covers involving $N$ by $\C_N$, i.e.
\[ \C_N = \{c = \{N_1, \dots, N_{b_c}\} \in \C : N=N_i \text{ for some } i \in [b_c]\}.\] 
Observe that since we assume $t \le n$, we have $|\C| = O(2^{n^2})$. 

We define an ILP $\L$ over the set of variables $x = (x_c : c \in \C)$ and $z = (z_N : N \subseteq U)$, with the following inequalities:
\begin{eqnarray}
 & \sum_{c \in \C} x_c \ge d & \label{eqn:sumd} \\
 & \sum_{N \subseteq U} z_N \le s & \label{eqn:bound-s}\\
 & \sum_{c \in \C_N} x_c \le |\F_N|-z_N & \text{ for every } N \subseteq U \label{eqn:sizes} \\
 & 0 \le z_N \le |\F_N| & \text{ for every } N \subseteq U \label{eqn:bounds-z} \\
 & 0 \le x_c \le d & \text{ for every } c \in \C \label{eqn:bounds-x} 
\end{eqnarray}
Observe that $\kappa(\L)$ is upper bounded by a function of $n$ only.  The idea behind this model is to represent a family $\S$ of at most $s$ sets by variables $z$ (by deciding how many sets to take for each family $\F_N$, $N \subseteq U$), and to represent the disjoint set covers by variables $x$ (by deciding how many set covers will be equal to $c \in \C$). Then, inequalities (\ref{eqn:sizes}) will ensure that the set covers do not intersect with the chosen family $\S$. 
However, while we would be able to solve $\L$ in FPT time parameterized by $n$ by using, \eg, Lenstra's ILP Theorem, the reader might realize that doing so would not solve \RDSCPshort directly. Nevertheless, the following result establishes the crucial link between this system and our problem.
\begin{lemma}\label{lemma:equivsystem}
The \RDSCPshort instance is positive if and only if $\L$ is $z$-resilient.
\end{lemma}
\begin{proof}
Let us denote by $\L_z$ the ILP consisting only of inequalities involving variables $z$, \ie inequalities (\ref{eqn:bound-s}) and (\ref{eqn:bounds-z}).
Suppose first that the instance is positive (\ie, there exists $d$ disjoint set covers, each of size at most $t$ and disjoint from $\S$, for any $\S \subseteq \F$ of size at most $s$), and let $\sigma_z$ be an integral assignment for $z$ such that $\sigma_z$ satisfies $\L_z$.

We now define a family $\S \subseteq \F$ by picking, in an arbitrary manner, $\sigma_z(z_N)$ sets in $\F_N$, for each $N \subseteq U$ (since $\sigma_z(z_N) \le \min \{s, |\F_N|\}$, such a set $\S$ must exist). Since $\S$ is a family of at most $s$ sets, there exist $d$ disjoint set covers $V=\{\T_1, \dots, \T_d\}$ such that $\left(\bigcup_{i \in [d]} \T_i \right) \cap \S = \emptyset$. 
Then, for each $c \in \C$, let $\sigma_x(x_c)$ be the number of set covers of $U$ being equal to $c$. 
Clearly we have $\sigma_x(x_c) \in \{0, \dots, d\}$ and $\sum_{c \in \C} \sigma_x(x_c) = d$, and thus inequalities (\ref{eqn:sumd}) and (\ref{eqn:bounds-x}) are satisfied. 
Then, for all $N \subseteq U$, we may assume w.l.o.g. that $|\T_i \cap \F_N| \le 1$ for all $i \in \{1, \dots, d\}$. Hence $\sum_{c \in \C_N} \sigma_x(x_c)$ equals $|\left( \bigcup_{i \in [d]} \T_i \right) \cap \F_N|$, which is the number of sets of $\F_N$ involved in some set covers of $U$. 
Since $\left( \bigcup_{i \in [d]} \T_i \right) \cap \S = \emptyset$, we have $\sum_{c \in \C_N} \sigma_x(x_c) \le |\F_N| - \sigma_z(z_N)$, and thus inequalities (\ref{eqn:sizes}) are also satisfied for every $N \subseteq U$. Consequently, $\sigma_x \cup \sigma_z$ satisfies $\L$.

Conversely, let $\S \subseteq \F$, $|\S| \le s$. For each $N \subseteq U$, define $\sigma_z(z_N) = |\S \cap \F_N|$, which is thus an integral assignment of variables $z$ satisfying $\L_Z$. 
Hence, there exists a valid assignment $\sigma_x$ such that $\sigma_z \cup \sigma_x \models \L$. Then, for $c=\{N_1, \dots, N_b\} \in \C$, $b \le t$, consider a family of sets $\T$ consisting of a set chosen arbitrarily in $\S_{N_i}$ for each $i \in [b]$. By definition of $\C$, $\T$ is a set cover of size at most $t$. Then, since for all $N \subseteq U$, we have, by inequalities (\ref{eqn:sizes}), that it is possible to construct $\sigma_x(x_c)$ pairwise disjoint such sets for each $c \in \C$, each having an empty intersection with $\S$. In other words, for every $\S \subseteq \F$, $|\S| \le s$, there exist $d$ disjoint set covers (thanks to inequality (\ref{eqn:sumd})), each of size at most $t$ and not intersecting $\S$. In other words, the instance is positive. 
\end{proof}

Since, as we observed earlier, $\kappa(\L)$ is bounded by a function of $n$ only, combining Lemma \ref{lemma:equivsystem} with Theorem~\ref{thm:ILPres}, we obtain the following:

\begin{theorem}\label{thm:fpt-param-by-p-only}
\RDSCPshort is FPT parameterized by $n$.
\end{theorem}

\subsubsection{Other Results of  \RDSCPshort}\label{sec:other}
First observe that \RDSCPshort with $s=0$ and $d=1$ is exactly the classical \textsc{Set Cover} problem, which is NP-hard and even W[2]-hard parameterized by the size of the set cover (\ie $t$ in our case). This explains the W[2]-hardness with parameters $(s, d, t)$ and para-NP-hardness with parameters $(s, d)$.

Then, let us start from the following simple result. 
\begin{proposition}
\RDSCPshort is in XP when parameterized by $(s, d, t)$.
\end{proposition}
\begin{proof}
There are at most ${m \choose s}$ choices for $\S$ of size $s$ and for each such a choice there are at most $$d{m-s \choose t}$$ choices for $\T_1,\dots ,\T_d$, with $|\T_i| \leq t$ for all $i \in [d]$. For each such a choice we can decide whether $\T_1, \dots, \T_d$ are disjoint and whether each of them is a set cover in polynomial time. The result follows.
\end{proof}
Despite its simplicity, this result is actually somewhat tight as, as we will see, any smaller parameterization leads to a para-NP-hard or para-coNP-hard problem.

We now show that the problem is coNP-hard when $d=1$ and $ t=\tau$  for every fixed $\tau \ge 3$, implying para-coNP-hardness of  \RDSCPshort parameterized by $(d, t)$. 

\begin{theorem}\label{thm:paranphard-hittingset}
If $d=1$ and $t=\tau$, \RDSCPshort  is coNP-hard for every fixed $\tau \ge 3$.
\end{theorem}
\begin{proof}
We reduce from the \textsc{$\delta$-Hitting Set} problem, in which we are given a ground set $V=\{v_1, \dots, v_n\}$, a set $S = \{S_1, \dots, S_m\}$ with $S_j \subseteq V$ and $|S_j|=\delta$ for all $j \in [m]$ and an integer $k$, and where the goal is to decide whether there is a set $C \subseteq V$ of size at most $k$ and such that $C \cap S_j \neq \emptyset$ for all $j \in [m]$. This problem is known to be NP-hard for every $\delta \ge 2$ \cite{GaJo79}. 

Hence, let $\I = (V, S, k)$ be an instance of \textsc{$\delta$-Hitting Set} defined as above. For every $j \in [m]$, fix an arbitrary ordering of $S_j$, which can thus be seen as a tuple $(v_{i_1}, \dots, v_{i_{\delta}})$, allowing us to define $S_j[x] = v_{i_x}$ for all $x \in [\delta]$. 

We now define an instance $\I'$ of \RDSCPshort with universe $U$ and family $\F$ of sets. The universe consists of three parts: $U = U^V \cup U^S \cup \{u^*\}$, where $U^S = \bigcup_{j=1}^m P^j$ with $P^j = \{p_1^j, \dots, p_{\delta}^j\}$ for every $j \in [m]$, and $U^V$ contains one elements $u^V_Q$ for every subset $Q$ of $\delta-1$ elements among $[n]$. The family $\F$ of sets consists of two parts: $\F = \F^V \cup \F^S$, where $\F^V = \{s_1^V, \dots, s_n^V\}$ and $\F^S = \{s^S_1, \dots, s^S_m\}$. We now define the elements each set is made of. For every $i \in [n]$, the set $s^V_i$ consists of $\{p^j_x : j \in [m], x \in [\delta]$ such that $S_j[x] = v_i\}$ together with all elements $u^V_Q$ such that $i \notin Q$. For all $j \in [m]$, $s^S_j$ consists of $u^*$ together with $U^S \setminus P^j$. To conclude the construction of $\I'$, we let $t = \delta +1$, $d=1$, and $s=k$. Clearly this reduction can be done in polynomial time.

The remainder consists in proving that every set cover of $I'$ of size at most $t=\delta+1$ is of the form $\T_j=\{s^V_{i_1}, \dots, s^V_{i_{\delta}}, s^S_j\}$ such that $S_j = \{v_{i_1}, \dots, v_{i_{\delta}}\}$. 
If this is true, then observe that since, for every $j \in [m]$, the set $s^S_j$ only belongs to the set cover $\T_j$, we will be able to suppose \emph{w.l.o.g.} that it does not belong to any set that would intersect every set cover. Thus the set of set covers of $I'$ will be in one-to-one correspondence with the sets in $I$, implying that the obtained instance contains a set of size $s=k$ intersecting all set covers if and only if there is a hitting set of size at most $k$.

Let $\T \subseteq \S$ of size at most $t$. By construction, we need at least $\delta$ sets from $\F^V$ to be able to include all elements from $U^V$ (indeed, every set $B$ of $\delta-1$ sets of $\F^V$ corresponds to a subset $Q_B$ of $[n]$, and thus $B$ is only able to cover $U^V \setminus \{u^V_{Q_B}\}$), and we also need at least one set from $\F^S$ to contain $u^*$. Hence, $|\T \cap \F^V| = \delta$ and $\T \cap \F^S = \{s^S_j\}$ for some $j \in [m]$. Now, notice that $s^S_j$ contains all elements in $U^S$ but $P^j$, which implies that $\T \cap \S^V$ must contain $P^j$. However, this can only happen if $\T \cap \S^V = \{s^V_{i_1}, \dots, s^V_{i_{\delta}}\}$, where $S_j = \{v_{i_1}, \dots, v_{i_{\delta}}\}$, concluding the proof.
\end{proof}

\newcommand{\TDM}{\textsc{$3$-Dimensional Matching}\xspace}
We also settle the case of  \RDSCPshort parameterized by $(s, t)$. 

\begin{theorem}\label{thm:paranphard-matching}
If $s=0$ and $t=4$, \RDSCPshort is NP-hard.
\end{theorem}
\begin{proof}
We reduce from the \TDM problem, in which we are given three sets $X$, $Y$ and $Z$ of $n$ elements each, a set $M \subseteq X \times Y \times Z$ of hyperedges, and an integer $k$. The goal is to find $M' \subseteq M$ with $|M'| \ge k$ such that for all $e, e' \in M'$ with $e \neq e'$, $e=(x, y, z)$, $e'=(x', y', z')$, we have $x \neq x'$, $y \neq y'$ and $z \neq z'$ (in that case, we will say that these two hyperedges are \emph{disjoint}). Let $X = \{x_1, \dots, x_n\}$, $Y = \{y_1, \dots, y_n\}$ $Z = \{z_1, \dots, z_n\}$, and $M = \{e_1, \dots, e_m\}$.
We then define the following universe:
\[
 U = \{u_1^X, \dots, u_m^X\} \cup \{u_1^Y, \dots, u_m^Y\} \cup \{u_1^Z, \dots, u_m^Z\} \cup \{u_X, u_Y, u_Z, u_*\}
\]
%
 The family  $\F$ of sets is comprised of $\F_X$, $\F_Y$, $\F_Z$ and $\F^*$, where, for all $\omega \in \{X, Y, Z, *\}$, let $\F_{\omega} = \{s^{\omega}_1, \dots, s^{\omega}_n\}$. For each hyperedge $e_j = \{x_{i_1}, y_{i_2}, z_{i_3}\}$, the set $s^X_{i_1}$ (resp. $s^Y_{i_2}$, $s^Z_{i_3}$) contains $u_j^X$ (resp. $u_j^Y$, $u_j^Z$) and $u_X$ (resp. $u_Y$, $u_Z$), and the set $s^*_j$ consists of all elements but $u_j^X$, $u_j^Y$, $u_j^Z$, $u_X$, $u_Y$ and $u_Z$. 
To conclude the construction, which can be done in polynomial time, we set $d = k$ (and recall that $t=4$).

First, suppose that there exists a solution $M'$ for the \TDM problem. Without loss of generality, assume that $|M'| = k$, $M' = \{e_1, \dots, e_k\}$, and that $e_i = (x_i, y_i, z_i)$ for all $i \in [k]$ (recall that all members of $M'$ are pairwise disjoint). Then, observe that for all $i \in [k]$, the set $s_i^*$ consists of all elements but $u_i^X$, $u_i^Y$, $u_i^Z$, $u_X$, $u_Y$ and $u_Z$. However, $s^X_i$ consists of $u_i^X$ and $u_X$, set $s^Y_i$ consists of $u_i^Y$ and $u_Y$, and set $s^Z_i$ is composed of $u_i^Z$ and $u_Z$. Hence, $s^X_i \cup s^Y_i \cup s^Z_i \cup s^*_i = U$, and, since all members of $M'$ are pairwise disjoint, we thus constructed $d$ disjoint set covers of size at most $4$ each.

Conversely, suppose that there exist $\T_1, \dots, \T_d$, pairwise disjoint subfamilies of $\F$ such that for all $i \in [d]$, we have $|\T_i| = 4$ and $\cup_{X \in \T_i} X = U$. We first claim that for all $i \in [d]$, $\T_i$ intersects $\F_X$ (resp. $\F_Y$, $\F_Z$ and $\F_*$) on exactly one set.
Indeed, otherwise, since $|\T_i| = 4$ and since all sets in $\F_X$ (resp. $\F_Y$, $\F_Z$, $\F_*$) only contains $u_X$ (resp. $u_Y$, $u_Z$, $u_*$) among $\{u_X, u_Y, u_Z, u_*\}$, $\T_i$ would not be able to cover all $U$. Thus, we know that for all $i \in [d]$, we have $\T_i = \{s^X_{i_1}, s^Y_{i_2}, s^Z_{i_3}, s^*_{i_4}\}$, for some $(i_1, i_2, i_3, i_4) \in [n] \times [n] \times [n] \times [m]$. We claim that $(x_{i_1}, y_{i_2}, z_{i_3}) = e_{i_4}$. Indeed, observe that the set $s^*_{i_4}$ contains all elements but $u_{i_4}^X$, $u_{i_4}^Y$, $u_{i_4}^Z$, $u_X$, $u_Y$ and $u_Z$. By construction, the only way for $\T_i$ to cover $U$ is that set $s_{i_1}^X$ (resp. $s_{i_2}^Y$, $s_{i_3}^Z$) contains $u_{i_4}^X$ (resp. $u_{i_4}^Y$, $u_{i_4}^Z$) or, in other words, that $x_{i_1}$ (resp. $y_{i_2}$, $z_{i_3}$) belongs to hyperedge $e_{i_4}$. Thus, there exist $k$ pairwise disjoint hyperedges in $M$.
\end{proof}

\section{Closest String Problem}\label{sec:closeststring}

In the \textsc{Closest String} problem, we are given a collection of $k$ strings $s_1, \dots, s_k$ of length $L$ over a fixed alphabet $\Sigma$, and a non-negative integer $d$. The goal is to decide whether there exists a string $s$ (of length $L$) such that $d_H(s, s_i) \le d$ for all $i \in [k]$, where $d_H(s, s_i)$ denotes the Hamming distance between $s$ and $s_i$ (the number of positions in which $s$ and $s_i$ differ). If such a string exists, then it will be called a \emph{$d$-closest string}.

It is common to represent an instance of the problem as a matrix $C$ with $k$ rows and $L$ columns (\ie where each row is a string of the input); hence, in the following, the term \emph{column} will refer to a column of this matrix.
As Gramm \etal \cite{GrNiRo03} observe, as the Hamming distance is measured column-wise, one can identify some columns sharing the same structure. Let $\Sigma = \{\varphi_1, \dots, \varphi_{|\Sigma|}\}$. 
Gramm \etal show~\cite{GrNiRo03} that after a simple preprocessing of the instance, we may assume that for every column $c$ of $C$, $\varphi_i$ is the $i^{th}$ character that appears the most often (in $c$), for $i \in \{1, \dots, |\Sigma|\}$ (ties broken \wrt the considered ordering of $\Sigma$). Such a preprocessed column will be called \emph{normalized}, and by extension, a matrix consisting of normalized columns will be called \emph{normalized}.
One can observe that after this preprocessing, the number of different columns (called \emph{column type}) is bounded by a function of $k$ only, namely by the $k^{th}$ Bell number $\B_k = O(2^{k \log_2 k})$. The set of all column types is denoted by $T$.
Using this observation, Gramm \etal~\cite{GrNiRo03} prove that \textsc{Closest String} is FPT parameterized by $k$, using an ILP with a number of variables depending on $k$ only, and then applying Lenstra's theorem.


The motivation for studying resiliency with respect to this problem is the introduction of experimental errors, which may change the input strings~\cite{Pe00}. While a solution of the \textsc{Closest String} problem tests whether the input strings are consistent, a resiliency version asks whether these strings will remain consistent after some small changes. In the \textsc{Resiliency Closest String} problem we allow at most $m$ changes to appear anywhere in the matrix $C$. To represent this, we simply use the Hamming distance between two matrices. 
 
\begin{center}
\fbox{
\begin{minipage}[c]{.9\textwidth}
\textsc{Resiliency Closest String (RCS)}\\
\underline{Input:} $C$, a $k \times L$ normalized matrix of elements of $\Sigma$, $d \in \N$,
$m \leq kL$.\\
\underline{Question:} For every $C'$, $k \times L$ normalized matrix of elements of $\Sigma$ such that the Hamming distance of $C$ and $C'$ is at most $m$, does $C'$ admit a $d$-closest string?
\end{minipage}}
\end{center}


Let $\#_t$ be the number of columns of type $t$ in $C$. For two types $t, t' \in T$ let $\delta(t,t')$ be their Hamming distance. Let $z_{t,t'}$, for all $t, t' \in T$, be a variable meaning ``how many columns of type $t$ in $C$ are changed to type $t'$ in $C'$'' (we allow $t=t'$). Thus we have the following constraints:
\begin{align}
\sum_{t' \in T} z_{t,t'} &= \#_t & \forall t \in T \label{eqn:global:eq1} \\
\sum_{t, t' \in T} \delta(t,t') z_{t,t'} &\leq m \label{eqn:global:eq2}
\end{align}
These constraints clearly capture all possible scenarios of how the input strings can be modified in at most $m$ places. Then let $\#'_t$ be a variable meaning ``how many columns of $C'$ are of type $t$'', and let $x_{t,\varphi}$ represent the number of columns of type $t$ in $C'$ whose corresponding character in the solution is set to $\varphi$. Finally let $\Delta(t, \varphi)$ be the number of characters of $t$ which are different from $\varphi$. As the remaining constraints correspond to our formulation of \textsc{ILP Resiliency}, we have:
\begin{align}
\sum_{t \in T} z_{t,t'} &= \#'_{t'} & \forall t' \in T \label{eqn:global:eq3}\\
\sum_{\varphi \in \Sigma} x_{t, \varphi} &=  \#'_t & \forall t \in T \label{eqn:global:eq4}\\
\sum_{t \in T} \sum_{\varphi \in \Sigma } \Delta(t, \varphi) x_{t, \varphi} &\le  d & \label{eqn:global:eq5}
\end{align}
This is the standard ILP for \textsc{Closest String}~\cite{GrNiRo03}, except that $\#'_t$ are now variables, and there exists a solution $x$ exactly when there is a string at distance at most $d$ from the modified strings given by the variables $\#'$. 
Let $\L$ denote the ILP composed of constraints (\ref{eqn:global:eq1}), (\ref{eqn:global:eq2}), (\ref{eqn:global:eq3}), (\ref{eqn:global:eq4}) and (\ref{eqn:global:eq5}). Finally, let $\cal Z$ denote variables $z_{t, t'}$ and $\#'_t$ for every $t, t' \in T$.
\begin{lemma}\label{lemma:grcs}
The \textsc{Resiliency Closest String} instance is satisfiable if and only if $\L$ is $\cal Z$-resilient.
\end{lemma}
\begin{proof}

Constraints involving variables $\cal Z$ are (\ref{eqn:global:eq1}), (\ref{eqn:global:eq2}) and (\ref{eqn:global:eq3}).
Suppose first that the instance is satisfiable, and let $\sigma_{\cal Z}$ be an integral assignment for $\cal Z$. We construct $C'$ from $C$ by turning, in an arbitrary way, $\sigma_{\cal Z}(z_{t, t'})$ columns of type $t$ to columns of type $t'$. By constraints (\ref{eqn:global:eq1}), $C'$ is well-defined, and by constraints (\ref{eqn:global:eq2}), the Hamming distance between $C$ and $C'$ is at most $m$. Then, by constraint (\ref{eqn:global:eq3}), matrix $C'$ contains $\sigma_{\cal Z}(\#'_t)$ columns of type $t$, for every $t \in T$. Since the instance is satisfiable, there exists a $d$-closest string $s$ of $C'$. For $t \in T$ and $\varphi \in \Sigma$, define $\sigma_x(x_{t, \varphi})$ as the number of columns of type $t$ in $C'$ whose corresponding character in $s$ is $\varphi$. Since, as we said previously, $C'$ has exactly $\sigma_{\cal Z}(\#'_t)$ columns of type $t$, constraint (\ref{eqn:global:eq4}) is satisfied for every $t \in T$. Then, since $s$ is a $d$-closest string for $C'$, constraint (\ref{eqn:global:eq5}) is also satisfied.

Conversely, suppose that $\L$ is $\cal Z$-resilient, and let us consider $C'$, a $k \times L$ normalized matrix of elements of $\Sigma$ such that the Hamming distance of $C$ and $C'$ is at most $m$. In polynomial time, we construct $\sigma_{\cal Z}(z_{t, t'})$ for every $t, t' \in T$ such that (\ref{eqn:global:eq1}) and (\ref{eqn:global:eq3}) are satisfied. By definition of $C'$, constraint (\ref{eqn:global:eq2}) is satisfied. Thus, there exists an integral assignment $\sigma_x$ satisfying (\ref{eqn:global:eq4}) and (\ref{eqn:global:eq5}). We now construct $s$ as a string having, for every column type $t \in T$ in $C'$, $\sigma_x(x_{t, \varphi})$ occurrence(s)(s) of character $\varphi$, for every $\varphi \in \Sigma$ (columns chosen arbitrarily among those of type $t$ in $C'$). Because of constraint (\ref{eqn:global:eq4}), and since $\#'_t$ is the number of columns of type $t$ in $C'$, $s$ is well-defined. Finally, observe that constraint (\ref{eqn:global:eq5}) ensures that $s$ is a $d$-closest string of $C'$, which concludes the proof.
\end{proof}

It remains to observe that for the above system of constraints $\cal L$, $\kappa({\cal L})$ is bounded by a function of $k$ (since $|T| = O(2^{k \log_2 k})$. We thus obtain the following result:

\begin{theorem}
\textsc{RCS} is $FPT$ parameterized  by $k$.
\end{theorem}

\section{Resilient Scheduling}\label{sec:scheduling}

A fundamental scheduling problem is makespan minimization on unrelated machines, where we have $m$ machines and $n$ jobs, and each job has a vector of processing times with respect to machines $p_j = (p_j^1, \dots, p_j^m)$, $j\in [n]$. If the vectors $p_j$ and $p_{j'}$ are identical for two jobs $j, j'$, we say these jobs are of the same \textit{type}. Here we consider the case when $m$ and the number of types $\theta$ are parameters and the input is given as $\theta$ numbers $n_1, \dots, n_\theta$ of job multiplicities. A \textit{schedule} is an assignment of jobs to machines. For a particular schedule, let $n_t^i$ be the number of jobs of type $t$ assigned to machine $i$. Then, the \textit{completion time} of machine $i$ is $C^i = \sum_{t \in [\theta]} p_t^i n_t^i$ and the largest $C^i$ is the \textit{makespan} of the schedule, denoted $C_{max}$. 

The parameterization by $\theta$ and $m$ might seem very restrictive, but note that when $m$ alone is a parameter, the problem is W[1]-hard even when the machines are identical (i.e., copies of the same machine) and the job lengths are given in unary~\cite{JansenKMS13}. Also, Asahiro et al.~\cite{AsahiroJMOZ07} show that it is strongly NP-hard already for \textit{restricted assignment} when there is a number $p_j$ for each job such that for each machine $i$, $p_j^i \in \{p_j, \infty\}$ and all $p_j \in \{1,2\}$ and for every job there are exactly two machines where it can run. Mnich and Wiese~\cite{MnichW14} proved that the problem is FPT with parameters $\theta$ and $m$.

%
%
%

A natural way to introduce resiliency is when we consider unexpected delays due to repairs, fixing software bugs, etc., but we have an upper bound $K$ on the total expected downtime. We assume that the execution of jobs can be resumed after the machine becomes available again, but cannot be moved to another machine, that is, we assume preemption but not migration. Under these assumptions it does not matter when specifically the downtime happens, only the total downtime of each machine. Given $m$ machines, $n$ jobs and $C_{max}, K \in \N$, we say that a scheduling instance has a \emph{$K$-tolerant makespan $C_{max}$} if, for every $d_1, \dots, d_m \in \N$ such that $\sum_{i=1}^m d_i \le K$, there exists a schedule where each machine $i \in [m]$ finishes by the time $C_{max}-d_i$. 
We obtain the following problem:

\begin{center}
\fbox{
\begin{minipage}[c]{.9\textwidth}
\textsc{Resiliency Makespan Minimization on Unrelated Machines}\\
\underline{Input:} $m$ machines, $\theta$ job types $p_1, \dots, p_\theta \in \N^m$, job multiplicities $n_1, \dots, n_\theta$, and $K, C_{max} \in \N$.\\
\underline{Question:} Does this instance have a $K$-tolerant makespan $C_{max}$ ?
\end{minipage}}
\end{center}

Let $x_t^i$ be a variable expressing how many jobs of type $t$ are scheduled to machine $i$. We have the following constraints, with the first constraint describing the feasible set of delays, and the subsequent constraints assuring that every job is scheduled on some machine and that every machine finishes by time $C_{max} - d_i$:
\begin{align*}
\sum_{i=1}^m d_i &\leq K & \\
\sum_{i=1}^m x_i^t &= n_t & \forall t \in [\theta] \\
\sum_{t=1}^\theta x_t^i p_t^i &\leq C_{max} - d_i & \forall i \in [m] 
\end{align*}
Theorem \ref{thm:ILPres} and the system of constraints above implies the following result related to the above-mentioned result of Mnich and Wiese~\cite{MnichW14}.

\begin{theorem}\label{thm:sch}
\textsc{Resiliency Makespan Minimization on Unrelated Machines} is FPT when parameterized by $\theta$, $m$ and $K$ and with $\max_{t\in [\theta],i\in [m]} p_i^t \le N$ for some number $N$ given in unary.
\end{theorem}
\begin{proof}
We recall the following ILP, denoted by $\L$:
\begin{align}
\sum_{i=1}^m d_i &\leq K & \label{eqn:sched:d} \\
\sum_{i=1}^m x_i^t &= n_t & \forall t \in [\theta] \label{eqn:sched:x}\\
\sum_{t=1}^\theta x_t^i p_t^i &\leq C_{max} - d_i & \forall i \in [m]  \label{eqn:sched:makespan}
\end{align}
We prove that the instance is satisfiable (\ie has a $K$-tolerant makespan $C_{max}$) if and only if $\L$ is $d$-resilient.
Suppose first that the instance is satisfiable, and let $\sigma_d$ be an integral assignment of variables $d_i$ satisfying constraint (\ref{eqn:sched:d}), that is, we have a scenario of delays $\sigma_d(d_1), \dots, \sigma_d(d_m)$ with total delay at most $K$. Thus, there exists a schedule where each machine $i \in [m]$ finished by time $C_{max}$ with expected delay $\sigma_d(d_i)$. By defining, for every machine $i \in [m]$ and every type $t \in [\theta]$, $\sigma_x(x_i^t)$ to be the number of jobs of type $t$ assigned to machine $i$, we obtain an integral assignment for variables $x_i^t$ satisfying constraints (\ref{eqn:sched:x}) and (\ref{eqn:sched:makespan}). That is, $\L$ is $d$-resilient.

Conversely, suppose that $\L$ is $d$-resilient, and let us consider a scenario of delays $d_1, \dots, d_m$ with total delay at most $K$, or, equivalently, an integral assignment $\sigma_d$ of variables $d$ satisfying constraint (\ref{eqn:sched:d}). Since $\L$ is $d$-resilient, there exists an assignment $\sigma_x$ of variables $x_i^t$ satisfying (\ref{eqn:sched:x}) and (\ref{eqn:sched:makespan}). Using the same arguments as above, there exists a schedule where each machine $i \in [m]$ finishes by time $C_{max}-d_i$.

Finally, observe that $\kappa(\L)$ is bounded by a function of $\theta$, $m$ and $K$ only.
\end{proof}

\section{Resilient Swap Bribery}\label{sec:bribery}

The field of computational social choice is concerned with computational problems associated with voting in elections. \textsc{Swap Bribery}, where the goal is to find the cheapest way to bribe voters such that a preferred candidate wins, has received considerable attention.  This problem models not only actual bribery, but also processes designed to influence voting (such as campaigning). It is natural to consider the case where an adversarial counterparty first performs their bribery, where we only have an estimate on their budget. The question becomes whether, for each such bribery, it is possible, within a given budget, to bribe the election such that our preferred candidate still wins.
The number of candidates is a well studied parameter~\cite{BredereckEtAl2015,DornSchlotter2012}. In this section we will show that the resilient version of \textsc{Swap Bribery} with unit costs (unit costs are a common setting, cf. Dorn and Schlotter~\cite{DornSchlotter2012}) is FPT using our framework. Let us now give formal definitions.

\noindent
\textbf{Elections.}
An election~$E = (C,V)$ consists of a set $C$ of $m$ candidates~$c_1,\ldots,c_m$ and a set~$V$ of voters (or votes).
Each voter $i$ is a linear order~$\pref_{i}$ over the set~$C$.
For distinct candidates~$a$ and $b$, we write $a\pref_i b$ if voter~$i$ prefers~$a$ over $b$.
We denote by $\textrm{rank}(c,i)$ the position of candidate~$c \in C$ in the order~$\pref_i$.
The preferred candidate is $c_1$.

\noindent
\textbf{Swaps.}
Let $(C,V)$ be an election and let $i \in V$ be a voter.
A \emph{swap} $\gamma = (a,b)_i$ in preference order~$\pref_i$ means to exchange the positions of $a$ and $b$ in $\pref_i$; denote the resulting order by $\pref_i^{\gamma}$; the \emph{cost} of $(a,b)_i$ is $\pi_i(a,b)$ (in the problem studied in this paper, we have $\pi_i(a, b) = 1$ for every voter $i$ and candidates $a, b$).
A swap~$\gamma=(a,b)_i$ is \emph{admissible in $\pref_i$} if $\rank(a,i) = \rank(b,i)-1$. 
A set $\Gamma$ of swaps is \emph{admissible in $\pref_i$} if they can be applied sequentially in~$\pref_i$, one after the other, in some order, such that each one of them is admissible.
Note that the obtained vote, denoted by $\pref_i^{\Gamma}$, is independent from the order in which the swaps of $\Gamma$ are applied.
We also extend this notation for applying swaps in several votes and denote it $V^\Gamma$.

\noindent
\textbf{Voting rules.}
A voting rule~$\mathcal{R}$ is a function that maps an election to a subset of candidates, the set of winners. We will show our example for rules which are scoring protocols, but following the framework of so-called ``election systems described by linear inequalities''~\cite{DornSchlotter2012} it is easily seen that the result below holds for many other voting rules. With a scoring protocol $s=(s_1, \dots, s_m) \in \N^m$, a voter $i$ gives $s_1$ points to his most preferred candidate, $s_2$ points to his second most preferred candidate and so on. The candidate with most points wins.

\begin{center}
\fbox{
\begin{minipage}[c]{.9\textwidth}
\textsc{Resiliency Unit Swap Bribery}\\
\underline{Input:} An election $E=(C,V)$ with each swap of unit cost and with a scoring protocol $s \in \N^m$, the adversary's budget $B_a$, our budget $B$.\\
\underline{Question:} For every adversarial bribery $\Gamma_a$ of cost at most $B_a$, is there a bribery $\Gamma$ of cost at most $B$ such that $E=(C,(V^{\Gamma_a})^\Gamma)$ is won by $c_1$?
\end{minipage}}
\end{center}

\begin{theorem}\label{thm:bribe}
\textsc{Resiliency Unit Swap Bribery} with a scoring protocol is FPT when parameterized by the number of candidates $m$.
\end{theorem}
\begin{proof} 
A standard way of looking at an election when the number of candidates $m$ is a parameter is as given by \textit{multiplicities} of \textit{voter types}: there are at most $m!$ total orders on $C$, so we count them and output numbers $n_1, \dots, n_{m!}$. Observe that for two orders $\pref,\pref'$, the admissible set of swaps $\Gamma$ such that $\pref' = \pref^{\Gamma}$ is uniquely given as the set of pairs $(c_i, c_j)$ for which either $c_i \pref c_j \wedge c_j \pref' c_i$ or $c_j \pref c_i \wedge c_i \pref' c_j$ (cf.~\cite[Proposition 3.2]{ElkindEtAl2009}). Thus it is possible to define the price $\pi(i,j)$, for $i,j \in [m!]$, of bribing a voter of type $i$ to become of type $j$ (since every swap is of unit cost, it does not depend on the users). Moreover, we can extend our notation $\rank(a,i)$ to denote the position of $a$ in the order of type $i$.

Similarly to our \textsc{Global Resiliency Closest String} approach, let $z_{ij}$, for all $i, j \in [m!]$, be a variable representing the number of voters of type $i$ bribed to become of type $j$, and let $y_i$, $i \in [m!]$, represent the election $E=(C,V^{\Gamma_a})$ after the first bribery. These constraints describe all possible adversarial briberies:

\begin{align*}
\sum_{j=1}^{m!} z_{ij} &= n_i & \forall i \in [m!] \\
\sum_{i=1}^{m!} z_{ij} &= y_j & \forall j \in [m!] \\
\sum_{i=1}^{m!} \sum_{j=1}^n \pi(i,j) z_{ij} &\leq B_a &
\end{align*}

The rest of the ILP is standard; variables $x_{ij}$ will describe the second bribery in the same way as $z_{ij}$ and variables $w$ will describe the election after this bribery, on which we will impose a constraint which is satisfied when $c_1$ is a winner:

\begin{align*}
\sum_{j=1}^{m!} x_{ij} &= y_i & \forall i \in [m!] \\
\sum_{i=1}^{m!} x_{ij} &= w_j & \forall j \in [m!] \\
\sum_{i=1}^{m!} \sum_{j=1}^{m!} \pi(i,j) x_{ij} &\leq B & \\
\sum_{k=1}^m \sum_{i: \rank(c_1,i)=k} w_i s_k &> \sum_{k=1}^m \sum_{i: \rank(c_j,i)=k} w_i s_k & \forall j=2, \dots, m
\end{align*}
The rest of the proof is similar to the proof of Lemma \ref{lemma:grcs}.
\end{proof}

\section{Discussion and Open Problems}\label{sec:dis}

For some time, Lenstra's theorem was the only approach in parameterized algorithms and complexity based on integer programming. Recently other tools based on integer programming have been introduced: the use of Graver bases for the $n$-fold integer programming problem~\cite{HeOnRo13}, the use of ILP approaches in kernelization \cite{EtKrMnRo17}, or, conversely, kernelization results for testing ILP feasibility \cite{JaKr15}, and an integer quadratic programming analog of Lenstra's theorem \cite{Lok15}. Our approach is a new addition to this powerful arsenal.

However, there still remain powerful tools from the theory of integer programming which, surprisingly, have not found applications in the design of parameterized algorithms.
For one example take a result of Hemmecke, K{\"o}ppe and Weismantel~\cite{HemmeckeKW:14} about \emph{2-stage stochastic integer programming with $n$ scenarios}.
They describe an FPT algorithm (which builds on a deep structural insight) for solving integer programs with a certain block structure and bounded coefficients.
For another example see the recent work of Nguyen and Pak~\cite{NguyenP:2017} which generalizes the problem solved by Eisenbrand and Shmonin.
They study the complexity of deciding \emph{short Presburger sentences} of the form $\forall x_1 \in P_1 \exists x_2 \in P_2 \cdots \exists x_k \in P_k: A(x_1, \dots, x_k) \leq b,$ where $P_1, \dots, P_k$ are polyhedra, $A$ is an integer matrix, $b$ is an integer vector and $x_1, \dots, x_k$ are required to be integer vectors.
A close reading of their result reveals that if $A$ and $b$ are given in unary, their algorithm is FPT for parameters $k$, the sum of dimensions of $x_1, \dots, x_k$, and the number of rows of $A$.

Another important research direction is the optimality program in parameterized algorithms pioneered by Daniel Marx.
Many of the prototypical uses of Lenstra's algorithm lead to FPT algorithms which have a double-exponential (i.e., $2^{2^{k^{O(1)}}}$) dependency on the parameter, such as the algorithms for \textsc{Closest String}~\cite{GrNiRo03} or \textsc{Swap Bribery}~\cite{DornSchlotter2012}.
Very recently, Knop et al.~\cite{KnopKM:2017} showed that many of these ILP formulations have a particular format which is solvable exponentially faster than by Lenstra's algorithm, thus bringing down the dependency on the parameter down to single-exponential.
This leads us to wonder what is the true complexity of, e.g., \textsc{Resiliency Closest String}?
All we can say is that the complexity of our algorithm is at best double-exponential but probably worse (depending on the complexity of parametric ILP in fixed dimension). Is this the best possible, or does a single-exponential algorithm exist?



\bibliographystyle{plain}

\bibliography{refs}

\begin{thebibliography}{10}

\bibitem{AsahiroJMOZ07}
Y.~Asahiro, J.~Jansson, E.~Miyano, H.~Ono, and K.~Zenmyo.
\newblock Approximation algorithms for the graph orientation minimizing the
  maximum weighted outdegree.
\newblock In {\em AAIM}, LNCS 4508, pages 167--177, 2007.

\bibitem{BredereckEtAl2015}
R.~Bredereck, P.~Faliszewski, R.~Niedermeier, P.~Skowron, and N.~Talmon.
\newblock Elections with few candidates: Prices, weights, and covering
  problems.
\newblock In {\em ADT'15}, LNCS 9346, pages 414--431, 2015.

\bibitem{chitnisTA11}
R.~Chitnis, M.~Cygan, M.~Hajiaghayi, and D.~Marx.
\newblock Directed subset feedback vertex set is fixed-parameter tractable.
\newblock {\em ACM Trans. Alg.}, 11(4):28:1--28:28, 2015.

\bibitem{CrGuKoWa17}
J.~Crampton, G.~Gutin, M.~Kouteck{\'{y}}, and R.~Watrigant.
\newblock Parameterized resiliency problems via integer linear programming.
\newblock In {\em CIAC}, LNCS 10236, pages 164--176, 2017.

\bibitem{cramptonAAIM11}
J.~Crampton, G.~Gutin, and R.~Watrigant.
\newblock A multivariate approach for checking resiliency in access control.
\newblock In {\em AAIM 2016}, LNCS 9778, pages 173--184, 2016.

\bibitem{CrGuWa16-sacmat}
J.~Crampton, G.~Gutin, and R.~Watrigant.
\newblock Resiliency policies in access control revisited.
\newblock In {\em Proc. SACMAT'16}, pages 101--111. ACM, 2016.

\bibitem{CyFoKoLoMaPiPiSa15}
M.~Cygan, F.~V. Fomin, L.~Kowalik, D.~Lokshtanov, D.~Marx, M.~Pilipczuk,
  M.~Pilipczuk, and S.~Saurabh.
\newblock {\em Parameterized Algorithms}.
\newblock Springer, 2015.

\bibitem{DornSchlotter2012}
B.~Dorn and I.~Schlotter.
\newblock Multivariate complexity analysis of swap bribery.
\newblock {\em Algorithmica}, 64(1):126--151, 2012.

\bibitem{DoFe13}
R.~G. Downey and M.~R. Fellows.
\newblock {\em Fundamentals of Parameterized Complexity}.
\newblock Texts in Computer Science. Springer, 2013.

\bibitem{EisenbrandS08}
F.~Eisenbrand and G.~Shmonin.
\newblock Parametric integer programming in fixed dimension.
\newblock {\em Math. Oper. Res}, 33(4):839--850, 2008.

\bibitem{ElkindEtAl2009}
E.~Elkind, P.~Faliszewski, and A.~Slinko.
\newblock Swap bribery.
\newblock In {\em Proc. SAGT 2009}, volume 5814 of {\em LNCS}, pages 299--310,
  2009.

\bibitem{EtKrMnRo17}
M.~Etscheid, S.~Kratsch, M.~Mnich, and H.~R{\"{o}}glin.
\newblock Polynomial kernels for weighted problems.
\newblock {\em J. Comput. Syst. Sci.}, 84:1--10, 2017.

\bibitem{FrTa87}
A.~Frank and {\'E}.~Tardos.
\newblock An application of simultaneous diophantine approximation in
  combinatorial optimization.
\newblock {\em Combinatorica}, 7(1):49--65, 1987.

\bibitem{GaJo79}
M.~R. Garey and D.~S. Johnson.
\newblock {\em Computers and Intractability: A Guide to the Theory of
  NP-Completeness}.
\newblock W. H. Freeman, 1979.

\bibitem{GrNiRo03}
J.~Gramm, R.~Niedermeier, and P.~Rossmanith.
\newblock Fixed-parameter algorithms for {CLOSEST} {STRING} and related
  problems.
\newblock {\em Algorithmica}, 37(1):25--42, 2003.

\bibitem{HemmeckeKW:14}
R.~Hemmecke, M.~K{\"o}ppe, and R.~Weismantel.
\newblock Graver basis and proximity techniques for block-structured separable
  convex integer minimization problems.
\newblock {\em Mathematical Programming}, 145(1-2):1--18, 2014.

\bibitem{HeOnRo13}
R.~Hemmecke, S.~Onn, and L.~Romanchuk.
\newblock $n$-fold integer programming in cubic time.
\newblock {\em Mathematical Programming}, 137(1):325--341, 2013.

\bibitem{JaKr15}
B.~M.~P. Jansen and S.~Kratsch.
\newblock A structural approach to kernels for {ILP}s: Treewidth and total
  unimodularity.
\newblock In {\em ESA}, pages 779--791, 2015.

\bibitem{JansenKMS13}
K.~Jansen, S.~Kratsch, D.~Marx, and I.~Schlotter.
\newblock {Bin packing with fixed number of bins revisited}.
\newblock {\em Journal of Computer and System Sciences}, 79(1):39--49, 2013.

\bibitem{Ka87}
R.~Kannan.
\newblock Minkowski's convex body theorem and integer programming.
\newblock {\em Math. Oper. Res.}, 12(3):415--440, August 1987.

\bibitem{Kannan90}
R.~Kannan.
\newblock Test sets for integer programs, $\forall \exists$ sentences.
\newblock In William~J. Cook and Paul~D. Seymour, editors, {\em Polyhedral
  Combinatorics, Proceedings of a {DIMACS} Workshop, 1989}, volume~1 of {\em
  {DIMACS} Series in Discrete Mathematics and Theoretical Computer Science},
  pages 39--48. {DIMACS/AMS}, 1990.

\bibitem{KnopKM:2017}
D.~Knop, M.~Kouteck{\`y}, and M.~Mnich.
\newblock Combinatorial $n$-fold integer programming and applications.
\newblock In {\em Proceedings of {AAMAS} 2018}, 2018.
\newblock to appear.

\bibitem{Len83}
H.~W. Lenstra.
\newblock Integer programming with a fixed number of variables.
\newblock {\em Mathematics of Operations Research}, 8(4):538--548, 1983.

\bibitem{LiWaTr09}
N.~Li, M.~V. Tripunitara, and Q.~Wang.
\newblock Resiliency policies in access control.
\newblock {\em {ACM} Trans. Inf. Syst. Secur.}, 12(4), 2009.

\bibitem{Lok15}
D.~Lokshtanov.
\newblock Parameterized integer quadratic programming: Variables and
  coefficients.
\newblock {\em CoRR}, abs/1511.00310, 2015.

\bibitem{MnichW14}
M.~Mnich and A.~Wiese.
\newblock Scheduling and fixed-parameter tractability.
\newblock {\em Mathematical Programming}, 154(1):533--562, 2014.

\bibitem{NguyenP:2017}
D~Nguyen and I~Pak.
\newblock Short {Presburger} arithmetic is in {P}.
\newblock {\em Proc. STOC 2017}, 2017.

\bibitem{NguPak17}
Danny Nguyen and Igor Pak.
\newblock Complexity of short {Presburger} arithmetic.
\newblock In {\em Proc. {STOC} 2017}, pages 812--820, 2017.

\bibitem{PananjadyBV15}
A.~Pananjady, V.~Kumar Bagaria, and R.~Vaze.
\newblock The online disjoint set cover problem and its applications.
\newblock In {\em 2015 {IEEE} Conference on Computer Communications, {INFOCOM}
  2015}, pages 1221--1229, 2015.

\bibitem{PananjadyBV17}
A.~Pananjady, V.Kumar Bagaria, and R.Vaze.
\newblock Optimally approximating the coverage lifetime of wireless sensor
  networks.
\newblock {\em {IEEE/ACM} Trans. Netw.}, 25(1):98--111, 2017.

\bibitem{Pe00}
P.~Pevzner.
\newblock {\em Computational Molecular Biology: An Algorithmic Approach}.
\newblock MIT Press, 2000.

\bibitem{Seb93}
A.~Seb\H{o}.
\newblock Integer plane multiflows with a fixed number of demands.
\newblock {\em J. Comb. Theory, Ser. {B}}, 59(2):163--171, 1993.

\end{thebibliography}


\end{document}